\documentclass[a4paper,russian,titlepage,12pt]{article}

\usepackage{babel,amsmath,amssymb,amsthm,epsfig}
\usepackage{graphicx}

\newcommand{\ssplane}{\mathrm{PL}}
\newcommand{\ffreach}{\mathrm{Reach}}

\newcommand{\ffreachp}{\mathrm{Reach}^{\text{p}}}
\newcommand{\ffreache}{\mathrm{Reach}}

\newcommand{\ffg}{\mathrm{G}}
\newcommand{\ffdepth}{\mathrm{d}}

\newcommand{\ffdir}{\mathrm{dir}}
\newtheorem{theorem}{Theorem}

\newtheorem{stat}{Statement}

\newtheorem*{conseq}{Corollary}

\begin{document}

\begin{center}\LARGE{\textbf{Improved Monotone Circuit Depth Upper Bound for Directed Graph Reachability}} \end{center}
\begin{center}
\large{S. A. Volkov \\
(Moscow, 2008)}
\end{center}

%%% \section*{Глубина схем, распознающих достижимость в графе}

\subsection*{Abstract}

We prove that the directed graph reachability problem can be
solved by monotone fan-in $2$ boolean circuits of depth
$(1/2+o(1))(\log_2 n)^2$, where $n$ is the number of nodes. This
improves the previous known upper bound $(1+o(1))(\log_2 n)^2$.
The proof is non-constructive, but we give a constructive proof of
the upper bound $(7/8+o(1))(\log_2 n)^2$.

\subsection*{Definitions}

By $[x]$ denote the floor of $x$.

For real-value functions $f,g$ we will write $f\asymp g$, if
$f=O(g)$, and $g=O(f)$; we will write $f\thicksim g$, if $\lim
(f/g)=1$.

By $\ffg(g_{11},g_{12},\ldots,g_{1n},\ldots,g_{n1},g_{n2},\ldots,
g_{nn})$ denote the directed graph such that the set of vertices
of this graph is $\{1,\ldots,n\}$, and the adjacency matrix of
this graph is $\{g_{ij}\}.$

By $V(G)$ denote the set of vertices of the graph $G$. Similarly,
by $E(G)$ denote the set of edges of the graph $G$.

We say that the path $p$ has \emph{length} $l$, if $p$ has $l$
edges.

The graph $G'$ is called \emph{$l$-closure} of the graph $G$, if
$V(G')=V(G)$, and
$$E(G')=\{(i,j):\ \text{$G$ has a path from $i$ to $j$ with length $\leq l$}\}.$$

Suppose $G$ is a graph, $S\subseteq V(G)$; by $G_S$ denote the
graph such that $V(G_S)=S$, $E(G_S)=\{(i,j)\in E(G):\ i,j\in S\}$.

In the following definitions $\ffg$ is an abbreviated notation for
$\ffg(g_{11},\ldots,g_{nn})$.

Put
$$
\ffreach_n(g_{11},\ldots,g_{nn})=\begin{cases}1,\text{ if the
graph $\ffg$ has a path from $1$ to $n$,} \\
0\text{ in the other case},\end{cases}
$$
$$
\ffreache_{n,l}(g_{11},\ldots,g_{nn})=\begin{cases}1,\text{ if the
graph $\ffg$ has a path from $1$} \\ \quad\quad\text{to $n$ with length $l$,} \\
0\text{ in the other case},\end{cases}
$$
$$
\ffreachp_{n,l}(g_{11},\ldots,g_{nn})=\begin{cases}1,\text{ if the
graph $\ffg$ has a path from $1$ to $n$} \\ \quad\quad\text{with length $\leq l$,} \\
0,\text{ if there is no path from $1$ to $n$ in the graph $\ffg$,} \\
 \text{undefined in the other cases}.\end{cases}
$$

In this paper we consider boolean monotone circuits with fan-in
$2$, i.e., circuits over the basis $\{x\&y,\ x\vee y\}.$

By $\ffdepth(\Sigma)$ denote the depth of the circuit $\Sigma.$

The ordered family of sets $(S_1,\ldots,S_m)$ is called an
$(n,m,s,l,d)$-family, where $n$, $m$, $s$, $l$, $d$ are natural
numbers, if the following conditions hold:
\begin{enumerate}
\item $|\bigcup_{i=1}^m S_i|\leq n$; \item for any $i$ we have
$|S_i|\leq s$; \item For any $A\subseteq\{1,\ldots,m\}$ such that
$|A|\geq\frac{md}{l}$, we have $|\bigcup_{i\in A}S_i|\geq n-d+1$.
\end{enumerate}
Order is needed for calculation convenience only (in the proof of
statement \ref{stat_family}).

Suppose $M=(S_1,\ldots,S_m)$ is a family of sets; by $\bigcup M$
denote the set $|\bigcup_{i=1}^m S_i|$.

\subsection*{Proof of the upper bound}

\begin{stat}\label{statborodin}There exists a family of monotone circuits $\Sigma_{n,l}$ such that for any
$n,l$, $\Sigma_{n,l}$ realizes $\ffreache_{n,l},$ and
$$
\ffdepth(\Sigma_{n,l})=\log_2 n\cdot\log_2 l+O(\log n).
$$
\end{stat}
\begin{proof}
By repeated matrix squaring.
\end{proof}

\begin{stat}\label{stat_meet}
Let $M$ be an $(n,m,s,l,d)$-family, $A_0,\ldots,A_{l'}$ ($l'\leq
l$) be different elements of $\bigcup M$; then there exists a set
$S\in M$ and numbers $k\leq\frac{l}{d}$, $0=i_0<i_1<\ldots<i_k=l'$
such that $A_{i_1},\ldots,A_{i_{k-1}}\in S$, and for any $j$
($0\leq j<k$) we have $i_{j+1}-i_j\leq 2d$.
\end{stat}
\begin{proof}[Proof]
Without loss of generality we can assume that $l'\geq 2d+1$.
Suppose $k=\left[\frac{l'}{d}\right]$,
$B_i=\{A_{id},A_{id+1},\ldots,A_{id+d-1}\}$, $1\leq i\leq k-1$. We
now prove that there exists $S\in M$ such that for any $i$ ($1\leq
i\leq k-1$) we have $S\cap B_i\neq\varnothing$. Assume the
converse. Then for any set $B_i$ there exists an $S\in M$ such
that $S\cap B_i=\varnothing$. Then there exists a set $B_i$ such
that there exists at least $\frac{m}{k-1}\geq \frac{md}{l}$ sets
$S\in M$ such that $S\cap B_i=\varnothing$. Since $|B_i|=d$, we
see that the cardinality of the union of these sets S does not
exceed $n-d$. This contradicts the definition of an
$(n,m,s,l,d)$-family.

Take the set $S\in M$ such that for any $B_j$ we have $S\cap
B_j\neq\varnothing$. Also, take $i_1,\ldots,i_{k-1}$ such that
$A_{i_j}\in B_j\cap S$ ($1\leq j\leq k-1$). Put $i_0=0$, $i_k=l'$.
Obviously, the numbers $i_0,\ldots,i_k$ satisfiy the conditions of
the statement.
\end{proof}

\begin{stat}\label{stat_reachp_ind}
Let $M$ be an $(n,m,s,l,d)$-family; suppose a monotone circuit
$\Sigma'$ realizes $\ffreachp_{s+2,\left[\frac{l}{d}\right]}$;
then there exists a monotone circuit $\Sigma$ such that $\Sigma$
realizes $\ffreachp_{n,l}$, and
$$
\ffdepth(\Sigma)=\log_2 m+\log_2 n\cdot\log_2
d+\ffdepth(\Sigma')+f(n)
$$
for some fixed function $f(n)=O(\log n)$.
\end{stat}
\begin{proof}[Proof]
Without loss of generality , we can assume that all sets of the
family $M$ are subsets of $\{1,\ldots,n\}$. Adding $1$ and $n$ to
each set of $M$, we obtain the $(n,m,s+2,l,d)$-family $M'$.

Now we describe a construction of the circuit $\Sigma$. Let $G$ be
an input graph. Suppose $G'$ is the $2d$-closure of the graph $G$.
From statement \ref{statborodin} it follows that the adjacency
matrix of $G'$ can be obtained by a monotone circuit of depth
$\log_2 n\cdot\log_2 d+O(\log n)$.

For each set $S\in M'$ we construct a block $\Sigma_S$ which is a
clone of the circuit $\Sigma'$; it takes the adjacency matrix of
$G'_S$ to it's input ($\Sigma_S$ determines the existence of a
needed path in $G'_S$ from the vertex $1$ to the vertex $n$); if
$|V(G'_S)|<s$, then the input matrix of $\Sigma_S$ is expanded by
zeros.

The disjunction of all outputs of the circuits $\Sigma_S$, $S\in
M'$ is declared to be the output of $\Sigma$.

It's clear that $\ffdepth(\Sigma)$ satisfies the condition of this
statement. Also, it's obvious that $\Sigma$ is monotone. Let us
prove that $\Sigma$ realizes $\ffreachp_{n,l}$.

Suppose $G$ does not have a path from $1$ to $n$; in this case
$G'$ does not have a path from $1$ to $n$; therefore, all blocks
$\Sigma_S$, $S\in M'$ produce $0$ to their outputs; thus, $\Sigma$
produces $0$.

Suppose $G$ has a path from $1$ to $n$ of length $\leq l$; let
$A_0A_1\ldots A_{l'}$ be the shortest of these paths; from the
statement \ref{stat_meet} it follows that there exists a set $S\in
M'$ and numbers $k\leq \frac{l}{d}$, $0=i_0<i_1<\ldots<i_k=l'$
such that $A_{i_1},\ldots,A_{i_{k-1}}\in S$ and for any $j$
($0\leq j<k$) we have $i_{j+1}-i_j\leq 2d$; also, by construction,
$A_{i_0},A_{i_k}\in S$; from this it follows that $G'$ has a path
$A_{i_0}A_{i_1}\ldots A_{i_k}$; note that the length of that path
does not exceed $\left[\frac{l}{d}\right]$, and all vertexes of
that path are in $S$; it follows that $\Sigma_S$ produces $1$ to
it's output. Therefore, $\Sigma$ produces $1$ to it's output.
\end{proof}

\begin{stat}\label{stat_family}
If the natural numbers $n,m,s,l,d$, $d\leq n$ satisfy the
condition
$$
\frac{dm \ln m}{l}+d\ln n-\frac{smd^2}{nl}<0,
$$
then there exists an $(n,m,s,l,d)$-family.
\end{stat}
\begin{proof}[Proof]
Consider matrixes $\{a_{ij}\}$, $1\leq i\leq m$, $1\leq j\leq s$,
$1\leq a_{ij}\leq n$. To each matrix $\{a_{ij}\}$ assign the
ordered family of sets
$(\{a_{11},\ldots,a_{1s}\},\ldots,\{a_{m1},\ldots,a_{ms}\})$.

Suppose $M\subseteq\{1,\ldots,m\}$, $D\subseteq\{1,\ldots,n\}$;
then we say that the matrix $\{a_{ij}\}$ is an $(M,D)$-matrix, if
for any $i\in M$, $1\leq j\leq s$ we have $a_{ij}\notin D$. It's
clear that for any matrix $\{a_{ij}\}$ the following statement is
satisfied: the family of sets corresponding to $\{a_{ij}\}$ is an
$(n,m,s,l,d)$-family if and only if for any $M,D$ such that
$|M|=\lceil md/l\rceil$, $|D|=d$, the matrix $\{a_{ij}\}$ is not
an $(M,D)$-matrix.

Now we estimate the probability of the event that for the matrix
$\{a_{ij}\}$ there exist sets $M$, $D$, $|M|=\lceil md/l\rceil$,
$|D|=d$ such that $\{a_{ij}\}$ is an $(M,D)$-matrix. The number of
ways to select $M$ does not exceed $m^{\lceil md/l \rceil}$; the
number of ways to select $D$ does not exceed $n^d$; the
probability of the event that the matrix $\{a_{ij}\}$ is an
$(M,D)$-matrix (for fixed $M$, $D$) does not exceed
$(1-\frac{d}{n})^{s\lceil md/l\rceil}$; therefore, the estimated
probability does not exceed
$$
n^d\cdot m^{\lceil md/l\rceil}\cdot
\left(1-\frac{d}{n}\right)^{s\lceil md/l
\rceil}=n^d\cdot\left(m\left(\left(1-\frac{d}{n}\right)^{n/d}\right)^{sd/n}\right)^{\lceil
md/l \rceil}\leq $$ $$\leq n^d\cdot(m\cdot\exp(-sd/n))^{\lceil
md/l \rceil}\leq n^d\cdot(m\cdot\exp(-sd/n))^{md/l} = \exp
\left(\frac{dm \ln m}{l}+d\ln n-\frac{smd^2}{nl}\right)<1.
$$
The next to last inequality follows from
$$
\frac{dm \ln m}{l}-\frac{smd^2}{nl}<0,$$ i.e.,
$m\cdot\exp(-sd/n)<1$.

\end{proof}
\begin{conseq}
If $m=n$, $l<n$, $s>\frac{2n\ln n}{d}$, $d\leq n$, then there
exists an $(n,m,s,l,d)$-family.
\end{conseq}
\begin{proof}[Proof]
Indeed,
$$
\frac{dm \ln m}{l}+d\ln n-\frac{smd^2}{nl}=\frac{md}{l}\left(\ln
m-\frac{sd}{n}\right)+d\ln n=\frac{nd}{l}\left(\ln
n-\frac{sd}{n}\right)+d\ln n<$$ $$<\frac{nd}{l}\cdot(-\ln n)+d\ln
n=\left(1-\frac{n}{l}\right)d\ln n<0.
$$
\end{proof}

\begin{theorem}There exists a family of monotone circuits $\Sigma_{n,l}$ such that for any
$n,l$, $\Sigma_{n,l}$ realizes $\ffreachp_{n,l}$, and
$$
\ffdepth(\Sigma_{n,l})=\log_2 n\cdot\log_2 l-\frac{1}{2}(\log_2
l)^2+o((\log n)^2).
$$
\end{theorem}
\begin{proof}[Proof]
Now we describe the construction of the circuit $\Sigma_{n,l}$ by
given $n$, $l$. Without loss of generality we can assume that
$l<n$. Put
$$
d=\left[2^{\sqrt{\log_2 n}}\right],\quad k=\left[\log_d
l\right],\quad q=2\ln n+3,\quad m=n.
$$
Now we construct the monotone circuits
$\Sigma_k,\Sigma_{k-1},\ldots,\Sigma_0$ such that $\Sigma_i$
realizes $\ffreachp_{n_i,l_i}$ ($0\leq i\leq k$), and
$$
l_i=\left[\frac{l}{d^i}\right],\quad
n_i=\left[\frac{nq^i}{d^i}\right]
$$
($0\leq i\leq k$).

The circuit $\Sigma_k$ can be constructed using statement
 \ref{statborodin}, i.e.
 $$
 \ffdepth(\Sigma_k)=\log_2 n_k\cdot\log_2 l_k+O(\log n)=O(\log
 n\cdot \log d)=o((\log n)^2).
 $$

 Now we describe the construction of $\Sigma_i$ ($i<k$) under assumption that $\Sigma_{i+1}$ is already constructed. Let us prove that for any $i<k$ there exists an $(n_i,n_i,n_{i+1}-2,l_i,d)$-family. Indeed,
 $$
 n_{i+1}-2=\left[\frac{nq^{i+1}}{d^{i+1}}\right]-2>\frac{nq^{i+1}}{d^{i+1}}-3=\frac{nq^{i+1}-3d^{i+1}}{d^{i+1}}\geq\frac{nq^{i+1}-3nq^i}{d^{i+1}}=
 \frac{nq^i(q-3)}{d^{i+1}}=$$ $$=\frac{2\ln
 n}{d}\cdot\frac{nq^i}{d^i}\geq\frac{2\ln
 n}{d}\cdot\left[\frac{nq^i}{d^i}\right]=\frac{2\ln n}{d}\cdot
 n_i;
 $$
this and the corollary of statement \ref{stat_family} prove an
existence of the family of sets.

From statement \ref{stat_reachp_ind} and an existence of a family
it follows that we can construct $\Sigma_i$ such that
 $$
 \ffdepth(\Sigma_i)=\ffdepth(\Sigma_{i+1})+\log_2 n_i\log_2
 d+O(\log n).
 $$

 Therefore, taking into account $\Sigma_{n,l}=\Sigma_0$, we obtain
$$
\ffdepth(\Sigma)=\log_2 d\cdot \sum_{i=0}^{k-1}\log_2 n_i+o((\log
n)^2)=\log_2 d\cdot\sum_{i=0}^{k-1}(\log_2 n+i\log_2 q-i\log_2
d)+o((\log n)^2)=$$ $$=k\log_2 d\log_2 n+\frac{k(k-1)}{2}(\log_2
q-\log_2 d)\log_2 d+o((\log n)^2)=$$ $$=\log_2 l\cdot \log_2
n+\frac{1}{2}\left(\frac{\log_2 l}{\log_2 d}\right)^2(\log_2
q-\log_2 d)\log_2 d+o((\log n)^2)=$$ $$=\log_2 l\log_2
n-\frac{1}{2}(\log_2 l)^2\left(1-\frac{\log_2 q}{\log_2
d}\right)+o((\log n)^2)=\log_2 n\log_2 l-\frac{1}{2}(\log_2
l)^2+o((\log n)^2).
$$
\end{proof}
\begin{conseq}There exists a sequence of monotone circuits
$\Sigma_n$ such that for any $n$, $\Sigma_n$ realizes
$\ffreach_n$, and
$$
\ffdepth(\Sigma_n)\sim\frac{1}{2}(\log_2 n)^2.
$$
\end{conseq}

\subsection*{An explicit construction}

Now we describe an explicit construction that proves the upper
bound $\frac{7}{8}(\log_2 n)^2$.

Take a prime $q$. Suppose $Q=\mathrm{GF}(q)$.

The set $\{(x,y):\ x,y\in Q\}$ is called the \emph{plane}.

Elements of the plane are called \emph{points}.

For any $a,b,c\in Q$ such that $a\neq 0$ or $b\neq 0$, the set
$\{(x,y):\ x,y\in Q,\ ax+by+c=0\}$ is called a \emph{line}.

\begin{stat}\label{stat_lines}
Suppose $\overline{s}_1,\ldots,\overline{s}_l$ are distinct lines,
$|\bigcup_{i=1}^{l}\overline{s}_i|=q^2-u$; then we have
$$
l\leq \frac{(q+1)(q^2-u)}{u+q}.
$$
\end{stat}
\begin{proof}[Proof]
Now we give some additional definitions. Two different lines are
said to be \emph{parallel}, if they don't intersect. By
$\ffdir(\overline{l})$ we denote the set of all lines
$\overline{l}'$ such that $\overline{l}$ parallel to
$\overline{l}'$. The set $\ffdir(\overline{l})$ is called the
\emph{direction} of the line $\overline{l}$. Note that the number
of points in the plane is $q^2$, the number of lines in the plane
is $q(q+1)$, the number of directions is $q+1$. The set of all
directions is denoted by $\tilde{D}$.

Let $\overline{L}$ be the considered set of lines
($|\overline{L}|=l$), let $\overline{V}$ be the complement of
$\overline{L}$ to the set of all lines in the plane
($|\overline{V}|=v=q(q+1)-l$). Suppose
$P=\bigcup_{\overline{l}\in\overline{L}}\overline{l}$, $U$ is the
complement of $P$ to the plane ($|P|=p=q^2-u$, $|U|=u$).

For any line $\overline{l}$ and set of points $U$ by
$U_{\overline{l}}$ denote the set $U\cap \overline{l}$; put
$u_{\overline{l}}=|U_{\overline{l}}|$. For any direction
$\tilde{d}$ and set of lines $\overline{V}$ by
$\overline{V}_{\tilde{d}}$ denote the set of all lines
$\overline{l}$ such that $\overline{l}\in\overline{V}$ and
$\ffdir(\overline{l})=\tilde{d}$; put
$v_{\tilde{d}}=|\overline{V}_{\tilde{d}}|$.

By $R$ denote the set of all unordered pairs of points
$\{\alpha,\beta\}$ such that $\alpha,\beta\in U$ and $\alpha\neq
\beta$. Analogously, for any line $\overline{l}$ by
$R_{\overline{l}}$ denote the set of all unordered pairs of points
$\{\alpha,\beta\}$ such that $\alpha,\beta\in U_{\overline{l}}$
and $\alpha\neq \beta$.

Note that for any line $\overline{l}$ such that $\overline{l}\cap
U\neq\varnothing$, we have $\overline{l}\in\overline{V}$; also,
note that for any two different points there exists a unique line
that contains these points. From these facts it follows that the
set $\{R_{\overline{l}}:\ \overline{l}\in\overline{V}\}$ is a
partition of $R$. Therefore,
\begin{equation}\label{eq_lines_razb}
|R|=\sum_{\overline{l}\in\overline{V}}|R_{\overline{l}}|.
\end{equation}

Analogously, we see that for any direction $\tilde{d}$ the set
$\{U_{\overline{l}}:\ \overline{l}\in\overline{V}_{\tilde{d}}\}$
is a partition of the set $U$. Hence,
\begin{equation}\label{eq_lines_point_razb}
\sum_{\overline{l}\in\overline{V}_{\tilde{d}}}u_{\overline{l}}=u.
\end{equation}

Therefore, we have
$$
\frac{u(u-1)}{2}=|R|=\sum_{\overline{l}\in\overline{V}}|R_{\overline{l}}|=\sum_{\tilde{d}\in\tilde{D}}
\sum_{\overline{l}\in\overline{V}_{\tilde{d}}}|R_{\overline{l}}| =
\sum_{\tilde{d}\in\tilde{D}}
\sum_{\overline{l}\in\overline{V}_{\tilde{d}}}\frac{u_{\overline{l}}(u_{\overline{l}}-1)}{2}
\geq $$ $$\geq \frac{1}{2}\sum_{\tilde{d}\in
\tilde{D}}\left(\frac{u^2}{v_{\tilde{d}}}-u\right)\geq
\frac{1}{2}\left(\frac{(q+1)^2 u^2}{v}-u(q+1)\right),
$$
where the second equality follows from (\ref{eq_lines_razb}), the
first inequality follows from (\ref{eq_lines_point_razb}) and the
arithmetic-quadratic means inequality, the second inequality
follows from $\sum_{\tilde{d}\in\tilde{D}}v_{\tilde{d}}=v$ and the
arithmetic-harmonic means inequality.

Solving this inequality, we get
$$
v\geq\frac{u(q+1)^2}{u+q}.
$$
In other words,
$$
l\leq q(q+1)-\frac{u(q+1)^2}{u+q}=\frac{(q+1)(q^2-u)}{u+q}.
$$
\end{proof}

Now we describe an effective generation algorithm of
$(n,m,s,n,d)$-families by $n$, where $m\asymp n$, $s\asymp
n^{1/2}$, $d\asymp n^{3/4}$. Let $q$ be the minimal prime such
that $q^2\geq n$ (from Bertrand's postulate it follows that
$q\asymp n^{1/2}$). Suppose $d$ is the minimal natural number such
that
$$
\frac{d}{q}>\frac{q^2-d}{d+q}.
$$
It's clear that $d\asymp n^{3/4}$. By $\ssplane_q$ denote the
plane over the residual field by modulo $q$. From statement
$\ref{stat_lines}$ it follows that the set (arbitrary ordered) of
all lines in $\ssplane_q$ is a $(q^2,q(q+1),q,q^2,d)$-family. Take
some set $A\subseteq\ssplane_q$, $|A|=n$. Then the set
$$\{\overline{l}\cap A:\ \overline{l}\text{ is a line in
}\ssplane_q\}$$ (arbitrary ordered) is the needed family.

Now we describe an effective generation algorithm of the monotone
circuit $\Sigma_n$ by $n$, where $\Sigma_n$ realizes $\ffreach_n$,
and
$$
\ffdepth(\Sigma_n)\thicksim \frac{7}{8}(\log_2 n)^2.
$$

By the described algorithm we construct an $(n,m,s,n,d)$-family of
sets, where $m\asymp n$, $s\asymp n^{1/2}$, $d\asymp n^{3/4}$.
After this, using statement \ref{statborodin} we construct the
monotone circuit $\Sigma'_n$ such that $\Sigma'_n$ realizes
$\ffreachp_{s+2,[n/d]}$, and $\ffdepth(\Sigma'_n)\thicksim
\frac{1}{8}(\log_2 n)^2$. In the next place, using the constructed
family and statement \ref{stat_reachp_ind} (more precisely, the
proof of this statement), we obtain the needed circuit $\Sigma_n$.
It's clear that this algorithm works polynomial time (of n).
\end{document}